\documentclass[12pt,a4paper]{article}
\usepackage{amsmath,amsthm}
\usepackage[pdftex]{graphicx}
\usepackage{amssymb}
\usepackage{amsfonts}
\usepackage{mathpazo}
\usepackage[all,cmtip]{xy}
\usepackage{color}
\usepackage{tikz}
\usetikzlibrary{matrix}
\usepackage{pdfpages}
\usepackage{hyperref}
\usepackage[affil-it]{authblk}

\newtheorem{thm}{Theorem}[section]

\newtheorem{lemma}[thm]{Lemma}
\newtheorem{proposition}[thm]{Proposition}

\newtheorem{definition}[thm]{Definition}
\newtheorem{example}[thm]{Example}

\newtheorem{remark}[thm]{Remark}

\title{A post-quantum key exchange protocol from the intersection of quadric surfaces.}
\author{Daniele Di Tullio \hspace{50pt} Manoj Gyawali\thanks{This author is supported by INdAM Fellowship Programs in Mathematics and/or Applications cofunded by Marie Skłodowska-Curie Actions.}}
\providecommand{\keywords}[1]{\textbf{\textit{Keywords:}} #1}
\begin{document}
	\title{A post-quantum key exchange protocol from the intersection of quadric surfaces}
	\maketitle              
	\begin{center}
		\textit{Universit\`{a} degli Studi di Roma Tre, Department of Mathematics. Largo S. Leonardo Murialdo 1, Rome, Italy.}\\
		\href{mailto:danieleditullio@hotmail.it}{\small\texttt{ danieleditullio@hotmail.it}}\\
		\href{mailto:manoj.gyawali@ncit.edu.np}{\small\texttt{manoj.gyawali@ncit.edu.np}}
	\end{center}
	\begin{abstract}
		In this paper we present a key exchange protocol in which Alice and Bob have secret keys given by quadric surfaces embedded in a large ambient space by means of the Veronese embedding and public keys given by hyperplanes containing the embedded quadrics. Both of them reconstruct the isomorphism class of the intersection which is a curve of genus 1, which is uniquely determined by the $j$-invariant. An eavesdropper, to find this $j$-invariant, has to solve problems which are conjecturally quantum resistant.  
		
	\end{abstract}
	\section{Introduction}
	Bringing difficult mathematical problems to cryptography is required not only to connect abstract mathematics to the real world applications but also to make cryptography stronger and applicable. Many classical mathematical problems like factorization and discrete logarithm are  vulnerable to quantum attack after the algorithm by Shor \cite{shor} in 1994.
	The algorithm by Shor created a threat to the cryptographic world and then the necessity of the post-quantum system was realized. In 2016, the United States government agency National Institute of Standards and Technology (NIST) put a call for new post-quantum cryptographic algorithms to systematize the post-quantum candidates in near future  \cite{nist1} and in 2019 declared the 17 candidates for  public-key encryption and key-establishment algorithms and 9 candidates for digital signatures \cite{nist2} based on various mathematical problems. Currently, there are five major post-quantum areas of research are carried out, four of them are discussed in \cite{post} including lattice-based cryptography based on lattice problems, code-based cryptography based on decoding a generic linear code, which is an NP-complete problem \cite{code}, multivariate  cryptography based  on the difficulty of inverting a multivariate quadratic map or equivalently to solving a set of quadratic equations
	over a finite field which is an NP-hard problem, hash-based  cryptography  based on one way hash functions and isogeny based cryptography based on isogeny problems, see for ex. \cite{sidh,csidh}.\\
	In this paper, we propose a key exchange protocol whose security relies on various problems in computational algebraic geometry, like solving large system of polynomial equations with high degree in many variables, or finding the primary decomposition of an ideal generated by many polynomials in many variables, which we conjecture to be quantum-safe problems.\\
	In a nutshell: Alice chooses a quadric surface embedded in a large projective space by the means of the Segre and the Veronese map. She gives some information like an embedding and an automorphism of the variety so that Bob can generate an embedding which is required to agree on a common key. Both Bob and Alice have their respective embeddings by which they hide their secret quadric surfaces, instead they publish their corresponding hyperplanes containing the images of their respective embeddings. Now, by using their private embeddings they compute the pullback of each other's hyperplanes, recover a $ (2,2)$ homogeneous curve and finally compute the $ j $-invariant of the components. Under some heuristic assumptions, both parties are able to get such components with high probability. The $ j $-invariants are equal, which is the common keys for both Alice and Bob. Notwithstanding the availability of the public data, an attacker is not able to recover information on private data because of the assumptions on the underlying problems. \\
	
	In section 2  and section 3 we recall some terminologies that are used everywhere in this paper. In section 4 we give a key exchange protocol called Quadratic Surface Intersection (QSI) key exchange and a variant of it. In section 5 we present the QSI key exchange protocol in a scenario of a trusted third party. In section 6 we discuss some underlying mathematical problems and hardness assumptions. We also give appendices to fulfill some extra arguments.
	\section{Intersection of quadric surfaces}
	The intersection of two quadric surfaces in $ \mathbb{P}^{3} $ is a curve of degree 4. 
	The geometric properties of this curve are well known. In this section, we have taken most of the  terminologies from \cite{Ravi Vakil,Salmon,Shafarevich} unless otherwise stated. 
	\begin{proposition} \label{prop}
		Let $ \kappa $ be an algebraically closed field. Then, for a general choice of two quadric surfaces $ Q_{1},Q_{2}\subset\mathbb{P}^{3}_{\kappa} $, $ Q_{1} \cap Q_{2}$ is a smooth curve of genus 1. 
		\label{prop1}
	\end{proposition}
	The proposition \ref{prop} suggests that the intersection of two quadric surfaces is expected to be isomorphic to an elliptic curve, whose isomorphism class is determined by the $ j $-invariant. We describe a way to compute the $ j $-invariant of an intersection of two quadric surfaces. 
	\begin{definition}
		The standard Segre embeddings are a family of morphisms of projective variety
		\begin{center}
			\begin{tikzpicture}
			\node (1) at (0,0) {$\mathbb{P}^{n}\times\mathbb{P}^{m}$};
			\node (2) at (4,0) {$\mathbb{P}^{N_{n,m}}$};
			\node (3) at (-2,-0.5) {$([X_{0}: \ldots : X_{n}],	[Y_0: \ldots :Y_{m}])$};
			\node (4) at (5,-0.5) {$[X_{0}Y_{0}:...:X_{n}Y_{m}]$};
			\path[->]
			(1) 	edge node[above]{$s_{n,m}$}	(2);
			\path[|->]
			(3) edge (4);
			\end{tikzpicture}
		\end{center}
		where $ N_{n,m}=(m+1)(n+1)-1 $ and the sequence $ [X_{i}Y_{j}] $ is ordered by the standard lexicographical order. The images of these embeddings are called standard Segre varieties and they are denoted by the symbol $ \Sigma_{n,m} $. They are essentially isomorphic copies of $ \mathbb{P}^{n}\times\mathbb{P}^{m} $ inside $ \mathbb{P}^{N_{n,m}} $.
		\label{ definition standard Segre}
	\end{definition}
	\begin{example}
		$ \Sigma_{1,1} \subset \mathbb{P}^{3}$ is the smooth quadric surface defined by the equation
		\[
		X_{0}X_{3}=X_{1}X_{2}
		\]
		\label{Example Segre11}
	\end{example}
	Example \ref{Example Segre11} gives an easy characterization of the intersection of two smooth quadric surfaces $ Q_{1},Q_{2} $. We recall a basic result in Algebraic Geometry.
	\begin{lemma}
		All the smooth quadric hypersurfaces of $ \mathbb{P}^{n} $ are projectively isomorphic.
		\label{lemma isomorphic quadrics}
	\end{lemma}
	Suppose we have two smooth quadric surfaces $ Q_{1} $ and $ Q_{2}$. From lemma \ref{lemma isomorphic quadrics} and example \ref{Example Segre11} we can choose a projectivity $ {f:\mathbb{P} ^{3}\to \mathbb{P}^{3}}$ such that $ f(Q_{1})=\Sigma_{1,1} $. Assume that $ Q_{1}=\Sigma_{1,1} $, then $ s_{1,1}^{-1}(Q_{2})\cong Q_{1}\cap Q_{2} $. Let $ {F_{2}(Z_{0},Z_{1},Z_{2},Z_{3}) }$ be the quadratic form defining $ Q_{2} $, then $ {s_{1,1}^{-1}(Q_{2})} $ is defined in $ \mathbb{P}^{1}\times\mathbb{P}^{1} $ by a bi-homogeneous polynomial of bi-degree (2,2)
	\[
	G(X_{0},X_{1};Y_{0},Y_{1}):=F_{2}(X_{0}Y_{0},X_{0}Y_{1},X_{1}Y_{0},X_{1}Y_{1})
	\]
	which is called the "pullback" of the polynomial $ F_{2} $ through $ s_{1,1} $. Hence, in particular every intersection of two smooth quadric surfaces is isomorphic to the zero locus of a polynomial of bi-degree (2,2) in $ \mathbb{P}^{1}\times\mathbb{P}^{1} $. The next proposition explains how to compute the $ j $-invariant of a curve defined in that way.
	\begin{proposition}
		Let $ C\subset \mathbb{P}^{1}\times\mathbb{P}^{1} $ be a smooth curve defined by a bi-homogeneous polynomial of bi-degree (2,2) over a field of characteristic different from $ 0 $ and $ 3 $.
		\[
		F(X_{0},X_{1};Y_{0},Y_{1})=Y_{0}^{2}F_{0}(X_{0},X_{1})+Y_{0}Y_{1}F_{1}(X_{0},X_{1})+Y_{1}^2F_{2}(X_{0},X_{1})
		.\]
		Define $ G(X_{0},X_{1}):=F_{1}^{2}-4F_{0}F_{2} $ and write
		\[
		G(X_{0},X_{1})=q_{0}X_{0}^{4}+ q_{1}X_{0}^{3}X_{1}+q_{2}X_{0}^{2}X_{1}^{1}+q_{3}X_{0}X_{1}^{3}+q_{4}X_{1}^{4}
		\]
		Define
		\begin{eqnarray*}
			S &:=& q_{0}q_{4}-\dfrac{q_{1}q_{3}}{4} + \dfrac{q^{2}}{12}\\
			T &:=& \dfrac{q_{0}q2q_{4}}{6} +\dfrac{q_{1}q_{2}q_{3}}{48} - \dfrac{q_{2}^{3}}{216}
			- \dfrac{q_{0}q_{3}^{2}}{16}-\dfrac{q_{1}^{2}q_{4}}{16}
		\end{eqnarray*}
		Then $ j(C)=\dfrac{S^3}{S^3-27T^2} $.
	\end{proposition}
	\begin{proof}
		See appendix \ref{j invariant (2,2)-curve}.
	\end{proof}
	\section{Segre and Veronese embeddings}
	We recall here the general notion of Segre and Veronese embeddings. We already defined the standard Segre embeddings in  section 2. A general Segre embedding is a composition of the standard Segre embedding and a projective automorphism of the ambient space of the codomain, which is represented by a square matrix.
	\begin{definition}
		Let $ n,m\in \mathbb{N} $, $ N_{n,m}:=(m+1)(n+1)-1 $  and $ M\in \operatorname{GL}(N_{n,m}+1) $. Then we define
		\[
		s_{n,m}^{M}:=M\circ s_{n,m},\;\;\Sigma_{n,m}^{M}:=M\Sigma_{n,m}
		\]
		to be respectively the Segre embedding and the Segre Variety represented by the matrix $ M $. 
	\end{definition}
	\begin{remark}
		Since all the smooth quadric surfaces of $ \mathbb{P}^{3} $ are projectively isomorphic, then each of them is equal to some $ \Sigma_{n,m}^{M} $.
	\end{remark}
	\begin{example}
		Let
		\[
		M:=\begin{bmatrix}
		1 & -2 & 3 & 0\\
		0 & -1 & 1 & -5\\
		8 & 3 & 1 &-1\\
		1 & 0 & 0 & 1\\
		\end{bmatrix}
		\]
		be a matrix, then it represents the non-standard Segre embedding
		\begin{center}
			\begin{tikzpicture}
			\node (1) at (0,0) {$\mathbb{P}^{1}\times\mathbb{P}^{1}$};
			\node (2) at (6,0) {$\mathbb{P}^{3}$};
			\node (3) at (0,-1.5) {$\left(
				\begin{bmatrix}
				X_{0} \\ X_{1}
				\end{bmatrix},
				\begin{bmatrix}
				Y_{0}\\ Y_{1}
				\end{bmatrix}
				\right)$};
			\node (4) at (6,-1.5) {$\begin{bmatrix}
				X_{0}Y_{0}-2X_{0}Y_{1}+3X_{1}Y_{0}\\
				-X_{0}Y_{1}+X_{1}Y_{0}-5X_{1}Y_{1}\\
				8X_{0}Y_{0}+3X_{0}Y_{1}+X_{1}Y_{0}-X_{1}Y_{1}\\
				X_{0}Y_{0}+X_{1}Y_{1}
				\end{bmatrix}$};
			\path[->]
			(1) 	
			edge node[above]{$S_{1,1}^{M}$}	(2);
			\path[|->]
			(3) edge (4);
			\end{tikzpicture}
		\end{center}
	\end{example}
	We now define Veronese embeddings, which are copies of $ \mathbb{P}^{n} $ in a larger ambient space.
	\begin{definition}
		Let $ n,m\in\mathbb{N} $, then the standard Veronese embedding is the morphism
		\begin{center}
			\begin{tikzpicture}
			\node (1) at (0,0) {$\mathbb{P}^{n}$};
			\node (2) at (3,0) {$\mathbb{P}^{{n+m\choose m}-1}$};
			\node (3) at (-1,-0.5) {$[X_0:\dots:X_{n}]$};
			\node (4) at (4,-0.5) {$[X_{0}^{m}:\dots :X_{n}^{m}]$};
			\path[->]
			(1) 	edge node[above]{$v_{n,m}$}	(2);
			\path[|->]
			(3) edge (4);
			\end{tikzpicture}
		\end{center}
		where the sequence $[X_{0}^{m}:\dots :X_{n}^{m}]$ is ordered by the lexicographical order. The images of these embeddings are called standard Veronese varieties and they are denoted by $ V_{n,m} $.
	\end{definition}
	\begin{definition}
		Let $ n,m\in \mathbb{N} $, $ N_{n,m}:={n+m\choose m}-1 $ and $ M\in \operatorname{GL}(N_{n,m}+1) $. Then we define
		\[
		v_{n,m}^{M}:=M\circ v_{n,m},\;\;V_{n,m}^{M}:=MV_{n,m}
		\]
		to be respectively the Veronese embedding and the Veronese variety represented by the matrix $ M $. 
	\end{definition}
	For our purposes we are interested in Segre embeddings $ s_{1,1}^{M} $, the Veronese embeddings $ v_{3,m}^{M'}$. We give a name to their composition.
	\begin{definition}
		We call $ \sigma $-embedding any composition $ v_{3,m}^{M'}\circ s_{1,1}^{M} $.
	\end{definition}
	Any $ \sigma $-embedding is represented by a $ (N_{3,m}+1)\times(m+1)^2 $ matrix $ M $. It is defined by the condition
	\[
	\sigma([X_{0},X_{1}],[Y_{0},Y_{1}])=M\cdot \begin{bmatrix}
	X_{0}^{m}Y_{0}^{m}\\
	:\\
	X_{1}^{m}Y_{1}^{m}
	\end{bmatrix}
	\]
	
	We now describe how to construct automorphisms of the Veronese varieties. First of all we describe a natural multiplicative group homomorphism 
	\[
	\operatorname{GLEmb}(n,m):\operatorname{GL}(n+1)\to \operatorname{GL}(N_{n,m})
	\]
	arising from the standard Veronese embedding $ v_{n,m} $. Let $A:=(a_{ij})_{i,j\in\{0,...,n\}}\in\operatorname{GL}(n+1)$. It corresponds to an action on the coordinates
	\[
	X_{i}\mapsto L_{i}:=\sum_{j=0}^{n}a_{ij}X_{j},\;\;i\in\{0,..,n\}
	\]
	There is a natural action induced on the monomials of any fixed degree, in fact
	\[
	X_{0}^{e_{0}}\cdots X_{n}^{e_{n}}\mapsto L_{0}^{e_{0}}\cdots L_{n}^{e_{n}}
	\]
	We denote by $ \operatorname{GL}(n,m)(A) $ the matrix representing the action of $ A $ on the homogeneous polynomials of degree $ m $ with respect to the monomial basis with the standard lexicographical order.
	\begin{definition}
		We call general linear group embedding associated to the standard Veronese embedding $ v_{n,m} $ the function $ \operatorname{GLEmb}(n,m)$, which is defined above.
	\end{definition}
	\begin{example}
		In the case $ n=1,m=2,N_{n,m}=2 $ a general matrix 
		$ \begin{pmatrix}
		a & b\\
		c & d
		\end{pmatrix} \in \mathcal{GL}(2)$
		acts on the coordinates
		\begin{eqnarray*}
			X_{0} &\mapsto& aX_{0}+bX_{1}\\
			X_{1} &\mapsto& cX_{0}+dX_{1}
		\end{eqnarray*}
		Then the action on the monomials of degree 2 is the following
		\begin{eqnarray*}
			X_{0}^{2}&\mapsto&a^{2}X_{0}^{2}+2abX_{0}X_{1}+b^{2}X_{1}^{2}\\
			X_{0}X_{1}&\mapsto&acX_{0}^{2}+(ad+bc)X_{0}X_{1}+bdX_{1}^{2}\\
			X_{1}^{2}&\mapsto&c^{2}X_{0}^{2}+2cdX_{0}X_{1}+d^{2}X_{1}^{2}
		\end{eqnarray*}
		So $ \operatorname{GLEmb}(1,2)\begin{pmatrix}
		a & b \\
		c & d
		\end{pmatrix}=\begin{pmatrix}
		a^{2} & 2ab & b^2\\
		ac & ad+bc & bd\\
		c^2 & 2cd & d^2
		\end{pmatrix} $
	\end{example}
	The subgroup $ \operatorname{Im}(\operatorname{GLEmb}(n,m))\subset \operatorname{GL}(N_{n,m}+1) $ corresponds to the set of automorphisms of $ \mathbb{P}^{N_{n,m}} $ which fix $ V_{n,m} $.
	We can construct matrices representing automorphisms of any Veronese variety $ V_{n,m}^{M} $.
	\begin{proposition}
		$ \operatorname{Aut}(V_{n,m}^{M}) :=M\operatorname{Im}(\operatorname{GL}(n,m))M^{-1}$
	\end{proposition}	
	\begin{proof}
		It is a general fact
		\[
		\operatorname{Aut}(MX)=M\operatorname{Aut}(X)M^{-1}
		\]
		for any $ X\subset \mathbb{P}^{N}$ projective subvariety, for any $ M\in \operatorname{GL}(N+1) $.
	\end{proof}
	\section{Quadratic Surface Intersection (QSI) Key Exchange}
	In the proposed key exchange protocol both Alice and Bob choose  random quadric surfaces. The common key is the isomorphism class of the curve intersection of those quadrics, namely its $ j $-invariant. To make the exchange secure the quadric surface is embedded in a large projective space through a non-standard Veronese embedding.
	Any user $ U $ has a private data given by a non-standard Veronese embedding of $ \mathbb{P}^{3} $ represented by a $ N_{n,m}\times N_{n,m} $ matrix. The user has also a private data given by the isomorphic copy of quadric surface inside the chosen Veronese variety.
	$ U $ also needs to publish some data in order to allow anyone who wants to contact him to produce a distinct and random quadric surface: for this purpose he publishes some automorphisms of the Veronese variety and the user  chooses another quadric surface (distinct from the private one) inside it. These information should not allow any eavesdropper to  recover the Veronese embedding chosen by $ U $.
	\subsection{QSI algorithm first version.}
	The algorithm  is comprised of key generation and the key exchange. \\
	\textbf{User key construction:}
	\begin{enumerate}
		\item $ U $ chooses a finite field $ \mathbb{F}_{q} $.
		\item $ U $ chooses $ m\in \mathbb{N}^{+} $ and computes $ N={m+3 \choose 3} -1$.
		\item $ U $ chooses a non-standard Veronese embedding 
		$$ v_{3,m}^{M_{U}} :\mathbb{P}^{3}\to M_{U}\cdot V_{3,m}\subset\mathbb{P}^{N}$$
		represented by the matrix $ M_{U}\in \operatorname{GL}(N+1) $.
		\item  $ U $ constructs some automorphisms of $ M_{U}\cdot V_{3,m} $ by using the homomorphism $ \operatorname{GLEmb}(3,m) $. $ U $ chooses a set of automorphisms of $ \mathbb{P}^{3} $ of order $ q^{4}-1 $ (with a characteristic polynomial irreducible over $ \mathbb{F}_{q} $) ${\{U_{i}'\}_{1\leq i \leq t}\subset\operatorname{GL}(4)}$ and then  he computes
		\[
		U_i:=M_{U}\operatorname{GLEmb(3,m)(U_{i}')}M_{U}^{-1}
		\]
		we assume that $ t=2 $ is the most appropriate one.
		\item  $ U $ constructs a secret quadric surface inside $ M_{U}\cdot V_{3,m} $, more precisely a $ \sigma $-embedding
		\[
		\sigma^{(s)}_{U}:\mathbb{P}^{1}\times\mathbb{P}^{1}\to M_{U}\cdot V_{3,m}\subset \mathbb{P}^{N} 
		\]
		represented by a $ (N+1)\times (m+1)^{2}$ matrix $ M_{U}^{(s)} $. $U$ constructs also a hyperplane $ H_{U}\subset\mathbb{P}^{N} $ containing $ \operatorname{Im}(\sigma^{(s)}_{U}) $, which is represented by a vector in $\operatorname{coker}(M_{U}^{(s)})\subset \mathbb{F}_{q}^{N+1}$.
		\item  $ U $ constructs a public quadric surface inside $ M_{U}\cdot V_{3,m} $, more precisely a $ \sigma $-embedding
		\[
		\sigma^{(p)}_{U}:\mathbb{P}^{1}\times\mathbb{P}^{1}\to M_{U}\cdot V_{3,m}\subset \mathbb{P}^{N} 
		\]
		represented by a $ N+1\times (m+1)^{2}$ matrix $ M_{U}^{(p)} $.
	\end{enumerate}
	The key exchange is asymmetric since the common keys are different depending on if Alice wants to contact Bob or vice versa. Suppose that Bob wants to contact Alice.\\
	\textbf{Alice public keys:}
	\begin{itemize}
		\item The field $ \mathbb{F}_{q} $.
		\item $ m\in\mathbb{N}^{+} $.
		\item Two matrices $ A_{1},A_{2}\in \operatorname{GL}(N+1)$, where $ N={m+3\choose 3}-1 $.
		\item The $ (N+1)\times (m+1)^{2} $ matrix $ M_{A}^{(p)} $.
		\item The hyperplane $ H_{A}\in \mathbb{F}_{q}^{N+1} $.
	\end{itemize}
	\textbf{Alice secret keys:}
	\begin{itemize}
		\item The  matrix $ M_{A}^{(s)} $.
	\end{itemize}
	\textbf{Key Exchange:}
	\begin{enumerate}
		\item Bob chooses $ m_{1},m_{2},m_{1}',m_{2}'\in\{0,...,q^{4}-1\} $ and then computes\\ ${M_{B}':=A_{1}^{m_{1}}A_{2}^{m_{2}}A_{1}^{m_{1}'}A_{2}^{m_{2}'}}$.
		\item Bob computes the matrix $ M_{B}:=M_{B}'\cdot M_{A}^{(p)} $. This corresponds to a choice of a $ \sigma $-embedding $ {\sigma_{B}:\mathbb{P}^{1}\times\mathbb{P}^{1}\to\mathbb{P}^{N}} $.
		\item Bob computes a random $ H_{B}\in \operatorname{coker}(M_{B}) $ and sends it to Alice. This corresponds to a hyperplane containing $ \operatorname{Im}(\sigma_{B}) $.
		\item Bob computes the pullback $ \sigma_{B}^{*}H_{A} $. It is a curve in $ \mathbb{P}^{1} \times \mathbb{P}^{1} $ defined by a curve of bi-degree $ (m,m) $. He uses a factorization algorithm to find a component of bi-degree (2,2) then he computes its $ j $-invariant $ j_{B}\in \mathbb{F}_{q} $. The probability that the residue curve of bi-degree $ (m-2,m-2) $ is reducible is negligible (see appendix \ref*{Irreducibility (m,m) divisor} for more details), so the $ j_{B} $ is well determined except for $ m=4 $.
		\item Alice computes the pullback $ {\sigma_{A}^{(s)}}^{*}H_{B} $. She finds the component of bi-degree (2,2), then she computes its $ j $-invariant $ j_{A}\in\mathbb{F}_{q} $.
	\end{enumerate}
	$ j_{A}=j_{B} $ is the common key of Alice and Bob.
	\begin{example}(Toy Example)
		\textbf{Key Generation:}	
		A finite field $ \mathbb{F}_q$ with $ q = 67 $, $ m=3$. Alice chooses a random matrix 	
		\[ M_A = \left[ 
		\begin{smallmatrix} 
		56&21&22&46&19&30&54&59&17&23&35&17&18&60&13&54&27&43&55&16\\42&1&54&49&3&29&9&1&34&65&35&65&34&47&27&4&25&53&27&17\\29&3&59&56&50&44&36&27&63&33&3&15&4&36&28&32&3&50&29&56\\62&53&21&31&23&50&28&37&13&62&16&27&29&27&66&44&40&42&60&55\\9&2&58&4&50&2&63&34&10&1&27&34&14&17&11&61&4&48&36&61\\35&63&13&66&50&14&23&42&23&58&22&14&6&29&52&58&36&9&42&0\\61&64&15&65&57&44&54&60&11&4&25&37&29&5&56&9&9&11&57&31\\40&42&2&48&28&26&31&9&8&15&62&31&53&46&12&39&46&52&30&8\\29&28&40&29&36&62&32&57&47&25&11&10&55&8&39&43&3&66&64&29\\61&63&24&42&43&15&18&63&21&64&60&14&26&8&12&0&23&61&28&66\\55&21&48&0&47&17&20&22&61&65&49&5&24&43&51&24&62&14&41&42\\34&25&43&37&36&11&16&11&65&59&61&54&22&66&66&49&28&20&26&3\\0&38&53&16&44&46&21&64&54&5&39&2&64&28&61&30&53&1&34&58\\18&65&52&54&6&31&43&3&46&7&5&26&29&55&39&65&12&4&33&63\\49&59&43&54&46&14&16&4&30&40&29&1&48&59&22&2&8&14&30&33\\34&46&63&8&14&51&1&29&6&52&46&47&25&9&13&28&8&33&25&34\\48&31&6&62&34&49&16&60&32&21&55&22&2&23&35&20&62&0&64&15\\33&45&48&62&5&0&1&65&66&35&43&34&5&18&11&57&41&6&53&41\\12&24&28&36&4&18&31&34&1&21&65&13&1&31&43&9&23&43&66&13\\41&9&45&49&6&38&40&4&50&45&10&14&13&18&40&23&6&33&13&39\\
		\end{smallmatrix}
		\right]
		\]
		
		which represents a choice of a Veronese embedding 
		\[
		v_{A}:\mathbb{P}^{3}\to V_A \subset \mathbb{P}^{19}
		.\] She produces a random automorphism of the Veronese variety $ V_A $ which is given by a matrix 	
		\[ A_1 = \left[
		\begin{smallmatrix}
		60&55&21&17&53&5&30&53&30&25&36&1&40&46&14&36&27&7&54&54\\0&59&51&65&25&62&57&4&23&55&8&53&8&34&36&24&36&33&55&60\\56&34&7&35&10&39&18&64&62&49&23&10&41&28&0&12&1&52&51&51\\61&42&35&43&29&44&30&35&28&61&6&48&30&54&21&37&8&21&48&0\\37&20&42&4&23&25&26&47&3&46&31&49&13&12&40&21&10&66&12&38\\27&39&26&35&4&43&35&48&5&57&56&28&55&65&15&23&27&43&41&53\\54&9&7&15&21&21&8&57&40&22&50&55&29&46&39&42&44&21&39&32\\3&17&9&43&3&48&38&49&24&5&20&60&56&57&31&18&53&57&21&43\\63&20&4&44&2&30&16&26&21&23&40&42&6&28&26&23&40&62&54&1\\23&6&0&50&47&46&29&17&62&62&32&55&26&22&27&65&29&25&65&11\\66&16&65&60&8&34&21&43&49&64&3&20&37&49&1&9&58&39&20&5\\38&19&30&22&11&15&20&9&16&1&65&12&17&1&56&41&4&21&44&19\\15&16&16&4&26&61&62&16&6&36&33&33&30&2&35&56&65&59&33&46\\44&22&9&56&57&11&10&6&14&22&24&58&49&32&35&58&4&14&53&48\\65&26&24&24&48&57&14&44&0&39&18&45&35&19&21&59&56&63&1&3\\39&24&5&34&46&14&31&14&59&40&1&38&43&46&40&21&33&65&20&36\\54&55&41&7&17&58&13&30&1&66&53&41&15&47&66&65&58&34&0&41\\55&37&41&32&11&42&38&25&43&9&33&7&31&25&59&3&45&61&36&36\\39&26&46&66&50&5&52&14&2&3&15&25&22&27&65&4&56&58&27&60\\46&35&61&39&20&51&21&50&55&29&18&38&12&46&59&51&2&43&15&31
		\end{smallmatrix} \right]
		\] 	
		and she keeps the secret embedding 
		\[
		\sigma^{(s)}_{A}:\mathbb{P}^{1} \times \mathbb{P}^{1}\to \mathbb{P}^{19}
		\]
		represented by the matrix	
		\[ M_A^{(s)} = \left[
		\begin{smallmatrix}	
		9&17&23&59&50&10&11&36&64&56&27&16&40&62&8&6\\34&48&49&26&25&58&16&33&36&2&43&4&62&39&17&34\\3&49&33&24&26&64&66&46&63&17&61&49&14&43&65&23\\37&45&12&60&27&5&29&3&33&51&7&30&11&47&40&44\\54&57&19&25&13&14&37&23&24&39&21&41&58&41&65&64\\40&18&2&44&64&32&17&15&42&15&27&30&45&22&39&49\\46&10&38&59&31&36&1&28&18&51&46&19&5&8&31&26\\45&65&39&32&35&6&18&65&12&0&17&59&65&4&26&5\\4&14&13&27&43&58&63&66&41&57&39&12&30&43&25&7\\4&1&6&25&49&11&40&48&20&30&52&29&35&23&35&3\\31&5&63&25&63&2&20&62&32&13&43&24&18&14&40&14\\21&38&7&31&46&50&3&27&8&59&47&21&29&53&22&1\\9&55&21&31&48&5&20&66&9&33&28&0&45&25&7&48\\19&48&3&13&6&20&1&33&37&12&61&63&36&34&55&35\\21&9&62&15&20&7&24&62&64&9&30&31&1&46&62&60\\37&52&30&58&33&46&63&28&38&22&14&36&12&30&10&59\\33&36&36&55&66&15&15&56&17&56&62&21&7&39&20&29\\30&40&53&59&8&9&62&28&13&31&4&41&44&24&47&51\\50&41&5&5&42&40&62&20&36&59&2&41&23&62&25&42\\24&66&52&46&24&23&64&8&45&52&59&29&10&4&33&12 \\
		\end{smallmatrix} \right]. 
		\] For a public key, she computes a hyperplane 
		$ H_A $, which is a closed subscheme of projective space  $ \mathbb{P}^{19} $ over $ \mathbb{F}_q $ defined by:
		\begin{eqnarray*}
			&x_0 - 21x_3 - 15x_5 - 32x_6 + 16x_7 - 10x_8 + 5x_9 + 11x_{10} + 16x_{11} + x_{12} - 4x_{13} \\
			&- 28x_{14} - 20x_{15} + 18x_{16} + 8x_{17} + x_{18} + 32x_{19}
		\end{eqnarray*} 
		containing the image of $ \sigma^{(s)}_{A}$ and also  computes an embedding 
		\[
		\sigma^{(p)}_{A}:\mathbb{P}^{1} \times \mathbb{P}^{1}\to \mathbb{P}^{19} \] 
		represented by the matrix 	
		\[ M_A^{(p)} = \left[
		\begin{smallmatrix}	
		28&32&4&38&26&36&6&20&1&2&22&27&23&35&10&24\\29&3&48&59&19&11&17&23&10&61&12&50&21&17&46&35\\27&21&45&43&33&48&4&43&64&37&34&46&60&3&41&27\\4&45&65&55&6&40&11&30&41&31&19&23&42&41&11&1\\53&14&44&37&15&49&57&26&55&22&29&45&44&9&59&30\\45&60&22&66&26&27&60&60&54&63&62&64&4&18&44&18\\42&33&47&53&52&65&45&28&5&21&58&45&49&31&22&27\\30&13&6&6&37&40&38&51&2&43&61&6&52&4&48&34\\66&54&47&64&13&20&66&23&31&36&55&42&11&27&39&17\\29&28&31&44&54&9&60&44&64&1&59&12&38&41&57&32\\10&18&2&58&38&5&35&14&55&16&22&61&18&13&55&46\\28&53&39&66&55&0&46&21&7&49&30&1&60&15&37&63\\54&24&9&29&24&42&51&50&35&10&50&18&16&44&10&7\\64&24&63&33&49&14&47&35&33&30&59&4&20&24&66&1\\59&37&43&25&55&7&21&26&62&44&64&45&66&4&46&62\\7&47&2&13&0&40&21&1&11&7&56&14&60&41&21&62\\51&20&31&4&55&36&27&22&61&42&32&51&56&20&26&5\\36&46&36&42&59&65&33&58&43&20&40&50&15&7&11&63\\18&15&53&12&3&8&50&66&55&39&42&8&10&31&1&29\\66&34&55&50&26&57&3&35&30&41&39&55&35&40&40&27\\
		\end{smallmatrix} \right]
		\]

		Bob chooses a random integer  $m_1 = 70$ (for sake of brevity, in the present example, we have chosen only one automorphism instead of two), computes an automorphism $M_B=A_1^{m_1}$ of the variety $ V_A $ and then a $ \sigma $-embedding 
		
		\[
		\sigma_{B}:\mathbb{P}^{1} \times \mathbb{P}^{1}\to \mathbb{P}^{19} \] represented by the matrix 
		\[ M_B = \left[
		\begin{smallmatrix}	
		46&26&7&53&4&8&3&63&15&66&49&2&4&9&56&4\\61&24&1&50&17&19&48&28&38&45&44&35&36&49&4&48\\22&28&65&20&18&24&4&65&37&17&57&7&29&59&20&35\\58&21&5&64&66&8&17&20&27&17&14&50&53&40&28&18\\25&56&30&19&39&62&44&21&37&1&22&43&8&13&30&19\\43&14&16&44&60&60&21&10&48&65&53&20&46&12&35&44\\27&52&64&40&51&25&13&62&48&57&53&14&43&25&42&50\\38&56&61&42&26&42&29&20&23&34&56&34&29&60&25&24\\15&35&27&50&4&7&56&25&66&36&47&38&38&41&3&28\\10&38&43&46&56&37&28&48&44&51&23&52&5&48&22&50\\34&50&50&37&12&29&11&40&31&5&52&13&53&58&29&52\\37&20&40&22&49&37&58&0&8&19&26&34&52&35&19&53\\19&39&22&65&5&7&62&35&51&43&16&35&48&43&38&44\\60&3&49&48&44&9&39&26&39&50&24&53&52&4&0&43\\0&23&54&63&27&26&59&23&8&50&62&42&37&48&26&60\\5&14&33&28&36&24&17&18&15&47&30&49&37&34&55&16\\59&57&9&62&38&57&61&30&5&13&58&64&36&40&18&61\\38&48&39&43&40&9&35&51&14&40&60&61&50&2&25&22\\38&51&37&66&26&25&40&6&64&22&62&16&12&57&14&9\\21&15&48&53&36&20&29&25&18&18&2&0&62&64&36&63 \\
		\end{smallmatrix} \right]
		\]
		He also computes a hyperplane 
		$ H_B $, which is  again the closed subscheme of projective space of dimension over $ \mathbb{F}_q $ defined by:
		\begin{eqnarray*}
			&x_0 + 4400x_1 + 4433x_2 + 8909x_3 + 26482x_4 + 3162x_5 - 4113x_6 + 24289x_7 \\
			& -15946x_8 +4813x_9 
		\end{eqnarray*} 
		and is identifed by an element of coker$ (M_B) $.
		
		\textbf{Key Exchange:}\\
		Bob computes the pullback 
		
		\begin{equation*}
		\begin{split}
		\sigma_{B}^{*}H_{A} &=  12x_0^3x_2^3 + 2x_0^2x_1x_2^3+13x_0x_1^2x_2^3 - 
		7x_1^3x_2^3 + 26x_0^3x_2^2x_3 +  29x_0^2x_1x_2^2x_3\\
		&- 19x_0x_1^2x_2^2x_3 + 16x_1^3x_2^2x_3 + 29x_0^3x_2x_3^2 - 30x_0x_1^2x_2x_3^2 +
		9x_1^3x_2x_3^2 \\
		& - 26x_0^3x_3^3 + 
		12x_0^2x_1x_3^3 - 22x_0x_1^2x_3^3 - 24x_1^3x_3^3
		\end{split}
		\end{equation*}
		
		and finds a component of bi-degree (2,2) 
		\begin{equation*}
		\begin{split}
		C_1 &= -x_0^2x_2^2 + 6x_0x_1x_2^2 + x_1^2x_2^2 + 28x_0^2x_2x_3 - 12x_0x_1x_2x_3 + 30x_1^2x_2x_3 - 4x_0^2x_3^2 \\
		&- 24x_0x_1x_3^2 - 9x_1^2x_3^2
		\end{split}
		\end{equation*}
		and computes the j-invariant $ j_B = j(C_1) = 57 \in \mathbb{F}_q$.
		Alice computes the pullback 
		\begin{equation*}
		\begin{split}
		\sigma^{(s)^{*}_{A}}H_B  &= -32x_0^2x_1x_2^3 - 15x_0x_1^2x_2^3 + 24x_1^3x_2^3 - 7x_0^2x_1x_2^2x_3 - 16x_0x_1^2x_2^2x_3 \\
		&- 29x_0^3x_2x_3^2 + 19x_0^2x_1x_2x_3^2 - 11x_0x_1^2x_2x_3^2 - 16x_1^3x_2x_3^2 - 27x_0^3x_3^3 \\
		&-5x_0^2x_1x_3^3 - 7x_0x_1^2x_3^3 + 26x_1^3x_3^3
		\end{split}
		\end{equation*} 
		and finds a component of bi-degree (2,2) 
		\begin{equation*}
		\begin{split}
		C_2 &= 33x_0x_1x_2^2 + x_1^2x_2^2 - 23x_0x_1x_2x_3 - 2x_1^2x_2x_3 + 32x_0^2x_3^2 - 19x_1^2x_3^2
		\end{split}
		\end{equation*}
		and computes the j-invariant $ j_A = j(C_2) = 57 \in \mathbb{F}_q$, which is the common key.
	\end{example}
	This version of the algorithm has some practical issues: the memory required to reach a good level of security is remarkably larger than the one needed by some of the most practical existing post-quantum cryptosystems, like lattice based ones and SIDH. In fact each user has a public key comprehensive of matrices of large size over $ \mathbb{F}_{q} $. Another issue is the speed of the key exchange, in particular the bottleneck is the computation of products of big matrices.
	
	\subsection{QSI algorithm second version}
	In this section we describe some modifications of the previous algorithm  which significantly improve the speed of the key exchange. It is not clear that whether this makes the protocol prone to some attacks by the extra information revealed. 
	
	Experimental evidence shows that the computationally heavy part of the key exchange is the calculation of  $ M_{B}' $: it is required to compute powers of matrices of large size. By choosing $ A_1,A_{2} $ to be generalized permutation matrices, Alice can dramatically speed up the computations of this product of matrices. This can be achieved by choosing $ A_{1}',A_{2}' $ and $ M_{A} $ to be generalized permutation matrices. A drawback is that the order of $ A_{i} $ is bounded by $ 4(q-1) $. So in this context, Alice should use a much larger value of $ q $ to reach the same level of security.
	
	Besides the improvement in the speed of the key exchange, the main issue remains the size of the public key. The major contribution to this size is by the matrix $ M_{A}^{(p)} $, which requires around ${ {l}\cdot(m+1)^{2}\cdot {m+3 \choose 3}} $ bits, where $ l $ is the binary length of $ q $. For example in the case $ l=64 , m=8$ this values is $ 855360 $, which is unpractical. A problem is to reduce the size of $ M_{B}^{(p)} $, which can be achieved by taking sparse or small entry matrices.
	
	\section{QSI Key Exchange with TTP}
	In this section we describe a variation of the key exchange where a trusted third party is allowed. TTP are not used in the design of the most common key exchange protocols, on the other hand they are required in several real-life applications. The advantage of the TTP (which will be called "Trent") in the case of the QSI key exchange protocol is that it allows the users to have a high level of security with a considerably short public key size and less time in common key generation. \\
	\textbf{Trent secret data}:
	\begin{itemize}
		\item A Veronese variety $ V_{T}\subset\mathbb{P}^{{m+3\choose 3}-1} $ for $ m\in\mathbb{N} $.
	\end{itemize}
	\textbf{Trent public data}:
	\begin{itemize}
		\item A finite field $ \mathbb{F}_{q} $.
		\item A positive integer $ m $.
		\item A matrix $ M_T $ of size ${m+3\choose 3}\times(m+1)^{2} $ representing a $\sigma $-embedding of $ \mathbb{P}^{1}\times\mathbb{P}^{1} $ into $ V_{T} $. 
		\item Two $ {m+3\choose 3}\times{m+3\choose 3} $ matrices $ T_{1},T_{2} $ representing automorphisms of $ V_{T} $ of order $ q^{4}-1 $.
	\end{itemize}
	Suppose that a user $ U $ wants to register to Trent's key exchange system. Then $U$ has to:
	\begin{enumerate}
		\item Download Trent's public data.
		\item Choose random integers $ 1\leq m_{1}^{U},m_{2}^{U},m_{1}'^{U},m_{2}'^{U}\leq q^{4}-1 $ and to compute $ M_{U}:=T_{1}^{m_{1}^{U}}T_{2}^{m_{2}^{U}}T_{1}^{m_{1}'^{U}}T_{2}^{m_{2}'^{U}}$. This corresponds to the choice of a $ \sigma $-embedding $ {\sigma_{U}:\mathbb{P}^{1}\times\mathbb{P}^{1}\to\mathbb{P}^{{m+3 \choose 3}-1}} $ such that $ {\operatorname{Im}(\sigma_U)\subset V_{T}} $.
		\item Compute a random $ H_{U}\in\operatorname{coker}(M_{U}) $. This corresponds to the choice of a 	hyperplane containing $ \operatorname{Im}(\sigma_U) $.
	\end{enumerate}
	Suppose that Alice and Bob want to generate a common key, then:
	\begin{enumerate}
		\item Alice downloads  Bob's public key $ H_{B} $, she computes $ \sigma_{A}^{*}H_{B} $ and she finds a component of bi-degree $ (2,2) $. Then she computes its $ j $-invariant $ j_{A} $.
		\item Bob downloads  Alice's public key $ H_{A} $, he computes $ \sigma_{B}^{*}H_{A} $ and he finds a component of bi-degree $ (2,2) $. Then he computes its $ j $-invariant $ j_{B} $.
	\end{enumerate}
	Now, $ j_{A}=j_{B} $ is the common key of Alice and Bob.
	\begin{remark}
		$ H_{U} $ is the public key. Its binary length is $ {m+3\choose 3}\cdot l $, where $ l $ is the binary length of $ q $. For example, imposing $ {m+3\choose 3}-(m+1)^2-1  $ coefficients equal to 0 and one coefficient equal to $ 1 $, $ H_U $ can be described by $ l(m+1)^2 $ bits. For $ l=64 $ and $ m=8 $ it is equal to 5184 bits: shorter than in SIDH or NTRU at 128-bit security level.
	\end{remark}
	\section{Underlying mathemtical problems}
	Suppose that an eavesdropper Eve wants to break the protocol. Then she has the following possible options:
	\begin{enumerate}
		\item She can try to find explicitly the Veronese variety, i.e. the  $ {m+3\choose 3}\times{m+3\choose 3} $  matrix $ M_{U} $ in the version without TTP or $ M_{T} $ in the TTP version. Note that, in the case of TTP version, if Eve is able to solve this problem, then she is able to break any communication between two users of the Trent system.
		\item  She can try to find $ M_{U}^{(s)} $ in the version without TTP, or $ M_{U} $ in the version with TTP. Note that, in the case of TTP version, if Eve is able to solve this problem, then she is able to break any communication between $ U $ and other users of the Trent system.
		\item She can find the explicit equations of $ V_{A} =M_{A}V_{3,m}$ in the version without TTP, or of $ V_{T} $ in the version with TTP. Then she can try to attack the single communication between Alice and Bob by searching the primary components of $ V_{A}\cap H_{A}\cap H_{B} $ (without TTP) or $ V _{T}\cap H_{A}\cap H_{B}$ (with TTP).
	\end{enumerate}
	
	Suppose that Eve wants to follow the first option. Also, suppose that we are in the case of the TTP version (in the other case, the problem is completely analogous). A possible attempt is to find $ M_{T} $ by solving a system of polynomial equations: she writes $ M_{T} $ as a matrix of $ {m+3\choose 3}^2 $ unknowns. The condition that $ T_{i} $ is an automorphism of $ V_{T} $ means that
	\[
	T_{i}M_{T}=M_{T}\operatorname{GLEmb(3,m)}(A)
	\]
	for some matrix $ A\in\operatorname{GL}(4) $.
	If Eve eliminates the variables $ \{a_{ij}\} $, then she gets polynomial conditions of extremely high degree on $ m_{i,j} $ (note that\\ $ M_{T}\operatorname{GLEmb(3,m)}(A) $ is a matrix whose components are bi-homogeneous polynomials whose bi-degree is $ (1,m) $ in the set of variables $ \{m_{ij}\} $ and $ \{a_{ij}\} $).
	The condition that $ \sigma_{U} $ is a $ \sigma $-embedding such that $ \operatorname{Im}(\sigma_{U})\subset V_{T} $ means that
	\[
	\sigma_{U}=M_{T}\circ v_{3,m}\circ A \circ s_{1,1}
	\]
	for some $ A\in\operatorname{Aut}(\mathbb{P}^{3}) $. Note that, like above, the matrix $ M_{E} $ representing the $ \sigma $-embedding $ M_{T}\circ v_{3,m}\circ A \circ s_{1,1} $ is a matrix whose components are  bi-homogeneous polynomials whose bi-degree is $ (1,m) $ in the set of variables $ \{m_{ij}\} $ and $ \{a_{ij}\} $. If Eve eliminates the variables $ \{a_{ij}\} $, then she gets polynomial conditions of extremely high degree on $ \{m_{ij}\} $.
	
	Suppose that Eve chooses the second option to attack the system. Then she wants to find
	$  m_{1}^{U},m_{2}^{U},m_{1}'^{U},m_{2}'^{U} $ such that $$ H_{U}\in \operatorname{coker}\left(T_{1}^{m_{1}^{U}}T_{2}^{m_{2}^{U}}T_{1}^{m_{1}'^{U}}T_{2}^{m_{2}'^{U}}\right).$$ Since the product of matrices is non-commutative, it seems that methods similar to Pollard rho or baby-step giant-step are not possible in this case. Since the family of quadric surfaces of $ \mathbb{P}^{3} $ is a 9-dimensional projective space, using a brute force attack (just choosing random values of $ m_{1}^{U},m_{2}^{U},m_{1}'^{U},m_{2}'^{U} $), Eve should find a $ \sigma $-embedding $ \sigma_{E}$ such that $ \operatorname{Im}(\sigma_{E})=\operatorname{Im}(\sigma_{U}) $ in around $ q^{9} $ trials (instead of $ q^{16} $ as one would expect).
	
	Suppose that Eve wants to choose the third option: she has first to compute the polynomial equations defining $ V_{T} $. This is not a hard problem because $ V_{T} $ is defined by $ m(m^{2}-1)(m^3 + 12m^2 + 59m + 66) $ degree-2 homogeneous polynomials  by proposition \ref{prop appendix} and that  can be found by  the methods of linear algebra. After this, one needs to find the irreducible components of the variety $ V_{T}\cap H_{A}\cap H_{B} $, this corresponds to find the primary decomposition of the ideal generated by the quadratic polynomials defining $ V_{T} $ and the two linear polynomials defining respectively $ H_{A} $ and $ H_{B} $.
	\section*{Acknowledgements} We would like to thank Ankan Pal for his constructive suggestions on an earlier version of this paper.
	
	\appendix
	\section{Probability that a random curve of bi-degree (2,2) in $ {\mathbb{P}^{1}\times\mathbb{P}^1}$ is singular.}
	Let $ \kappa$ be an algebraically closed field, then a general curve $ C \subset \mathbb{P}^{1}\times\mathbb{P}^{1}$ of bi-degree $ (2,2) $ is non-singular. More precisely:
	\begin{proposition}
		Let \[
		\mathcal{S}:=\{a_{ij}X_{0}^{2-i}X_{1}^{i}Y_{0}^{2-j}Y_{1}^{j}=0: i,j \in \{0,1,2\}, a_{ij}\in \kappa\}
		\]
		be the set of curves of bi-degree (2,2) defined over $ \mathbb{K} $. Identify $C\in \mathcal{S} $ with its coefficients (up to scalar multiplication) $ [a_{ij}]\in\mathbb{P}^{8} $. Then the condition of being singular is closed in the Zariski topology of $ \mathbb{P}^{8} $, i.e. is defined by a set of homogeneous polynomial equations in $ [a_{ij}] $. 
	\end{proposition} 
	The above proposition states that singular curves are very few compared to the smooth ones. You may imagine sets defined by polynomial equations in $ \mathbb{R}^{n} $ or $ \mathbb{C}^{n} $: these sets have a smaller dimension than the one of the ambient space, so their measure is 0. A similar situation occurs for algebraically closed fields. If we consider curves defined over a finite field $ \mathbb{F}_{q} $ then the probability of being singular is not 0, but it should  decrease when $ q $ increases and it should be negligible when $ q $ is very large.
	\label{Smoothness of (2,2)-curve}
	\section{$ j $-invariant of a $ (2,2) $-curve in $ {\mathbb{P}^{1}\times\mathbb{P}^{1}} $}
	A standard result in the theory of algebraic curves is that there is a bijection
	\[
	\left\lbrace\begin{matrix}
	\text{genus 1 curves up}\\
	\text{to isomorphism.}
	\end{matrix}
	\right\rbrace\longleftrightarrow
	\left\lbrace\begin{matrix}
	\text{4-tuples of distinct points of }\mathbb{P}^{1}\\
	\text{up to automorphism.}
	\end{matrix} \right\rbrace
	\]
	see for example \cite[19.5]{Ravi Vakil}. Let $ [C] $ be an isomorphism class of genus 1 curves. Let $ {\pi:C\to\mathbb{P}^{1}} $ be any degree 2 morphism. Then the 4-tuple of points associated to $ C $ is the branch locus of $ \pi $, which are by definition the points  $ P\in\mathbb{P}^{1} $ such that $\#\pi^{-1}(P)=1$.
	\begin{example}
		Let $ E $ be the elliptic curve defined by the equation $ Y^{2}Z=f(X,Z) $, where $ f(X,Z)=(X-aZ)(X-bZ)(X-cZ) $, let
		\begin{center}
			\begin{tikzpicture}
			\node (1) at (0,0) {$E$};
			\node (2) at (3,0) {$\mathbb{P}^{1}$};
			\node (3) at (0,-0.5) {$[X:Y:Z]$};
			\node (4) at (3,-0.5) {$[X:Z]$};
			\path[->]
			(1) 	edge node[above]{$\pi$}	(2);
			\path[|->]
			(3) edge (4);
			\end{tikzpicture}
		\end{center}
		be the degree 2 map to $ \mathbb{P}^{1} $. The branch locus of $ \pi $ is the set ${ \{[1:0],[a:1],[b:1],[c:1]\}} $.
	\end{example}
	\begin{example}
		Let $ C \subset \mathbb{P}^{1}\times\mathbb{P}^{1}$ be a smooth curve of bi-degree (2,2) and let
		\begin{center}
			\begin{tikzpicture}
			\node (1) at (0,0) {$C$};
			\node (2) at (1.5,0) {$\mathbb{P}^{1}$};
			\node (3) at (0,-0.5) {$(P,Q)$};
			\node (4) at (1.5,-0.5) {$P$};
			\path[->]
			(1) 	edge node[above]{$\pi$}	(2);
			\path[|->]
			(3) edge (4);
			\end{tikzpicture}
		\end{center}
		be the first projection map.  Let
		\[
		F(X_{0},X_{1};Y_{0},Y_{1})=Y_{0}^{2}F_{0}(X_{0},X_{1})+Y_{0}Y_{1}F_{1}(X_{0},X_{1})+Y_{1}^{2}F_{2}(X_{0},X_{1})
		\]
		be the defining polynomial of $ C $. Then the branch locus of $ \pi $ is the set of points $ P=[p_{0},p_{1}] $ for which the equation
		\[
		F(p_{0},{p}_{1};Y_{0},Y_{1})=0\]
		has one (double) solution. Equivalently $ F(p_{0},p_{1},Y_{0},Y_{1}) $ is a quadratic binary form with vanishing discriminant. So $ [p_0:p_{1}] $ is a root of the binary quartic form
		\[
		G(X_{0},X_{1}):=F_{1}(X_{0},X_{1}) ^{2}-4F_{0}(X_{0},X_{1})F_{2}(X_{0}X_{1})
		\]
		The invariants of a binary quartic form under the action of $ \operatorname{GL}(2)$ is classically well known (see for example \cite[199,p.189]{Salmon}, \cite[2.2]{invariant binary quintic}, or \cite[10.2]{Dolgachev Invariant Theory}): if we write
		\[
		G(X_{0},X_{1})=q_{0}X_{0}^{4}+ q_{1}X_{0}^{3}X_{1}+q_{2}X_{0}^{2}X_{1}^{1}+q_{3}X_{0}X_{1}^{3}+q_{4}X_{1}^{4}
		\]
		and we define
		\begin{eqnarray*}
			S &:=& q_{0}q_{4}-\dfrac{q_{1}q_{3}}{4} + \dfrac{q_2^{2}}{12}\\
			T &:=& \dfrac{q_{0}q_{2}q_{4}}{6} +\dfrac{q_{1}q_{2}q_{3}}{48} - \dfrac{q_{2}^{3}}{216}
			- \dfrac{q_{0}q_{3}^{2}}{16}-\dfrac{q_{1}^{2}q_{4}}{16}
		\end{eqnarray*}
		then $ \dfrac{S^{3}}{S^{3}-27T^2} $ is the invariant of $ G $ under the action of $ \operatorname{GL}(2) $ or equivalently the invariant of the set of points given by the roots of $ G $ under the action of $ \operatorname{PGL}(2) $. This is equal to the $ j $-invariant of the curve $ C $.
	\end{example}
	\label{Irreducibility (m,m) divisor}
	\label{j invariant (2,2)-curve}
	\section{Irreducibility of curves of bi-degree $ (d,d) $.}
	The pullback of a hyperplane $H$ through a $\sigma$-embedding is a curve of bi-degree $(m,m)$ in $ \mathbb{P}^{1}\times\mathbb{P}^{1} $, which has a component of bi-degree $(2,2)$. An important task is to have a well defined key exchange is to know if the residual $(m-2,m-2)$ curve is irreducible or not. We can assume that this residual curve is randomly chosen among the curves of bi-degree $(m-2,m-2)$, so a general question is: what is the probability that a curve of bi-degree $(d,d)$ in $ \mathbb{P}^{1}\times\mathbb{P}^{1} $ is irreducible? 
	\section{Irreducible components of $V_T\cap H_A\cap H_B$ }
	The next proposition gives the implicit description of any Veronese variety as intersection of quadric hypersurfaces of the ambient space. Without loss of generality, we can suppose that it is the standard Veronese variety. In this section, some technical terms from algebraic geometry are used.
	\begin{proposition} \label{prop appendix}
		The Veronese variety $ V_{3,m} $ is an intersection of 
		\[ h_{m}:=m(m^{2}-1)(m^3 + 12m^2 + 59m + 66) 
		\] 
		linearly independent quadric hypersurfaces.
	\end{proposition}
	\begin{proof} First of all we need to compute 
		$ h^{0}(\mathcal{I}_{V_{3,m}}(2))$. Since  $ V_{3,m} $ is projectively normal, then
		\begin{eqnarray*}
			h^{0}(\mathcal{I}_{V_{3,m}}(2)) &=& h^{0}(\mathcal{O}_{\mathbb{P}^{N_{3,m}}}(2))-
			h^{0}(\mathcal{O}_{V_{3,m}}(2))\\
			&=& h^{0}(\mathcal{O}_{\mathbb{P}^{N_{3,m}}}(2))-h^{0}(\mathbb{O}_{\mathbb{P}^{3}}(2m))\\
			&=& \dfrac{1}{2}\left[	{m+3\choose 3}+1\right]{m+3\choose 3}- {2m+3\choose 3}
		\end{eqnarray*}
		which is equal to the desired value.
	\end{proof}
	\begin{example}
		For $ m=8 $ there are 12726 linearly independent quadric hypersurfaces containing $ V_{3,m} $. It is the condition in the linear system of quadric surfaces of $ \mathbb{P}^{N_{3,m}} $ of codimension 969.	
	\end{example}
	A possible approach to find the quadratic equations defining $ V_{T} $ is to generate $ d_m - h_{m} $ points of $V_T$, where $d_m$ is the dimension of the space of all quadric hypersurfaces, sufficiently random points inside $ V_{T} $: this can be done easily using the knowledge of the $ \sigma $-embedding and of some of its automorphisms. After that one can find a basis of the family of quadratic polynomials vanishing on those points. These quadratic polynomials generate the ideal $ I_{V_{T}} $.
	
	\begin{proposition}
		$ V_{3,m} \subset \mathbb{P}^{N_{3,m}}$ is a 3-dimensional projective variety of degree $ m^{3} $. 
	\end{proposition}
	\begin{proof}
		In general $\deg(V_{n,m})=m^n$, see for example \cite[4.2.7]{Shafarevich}
	\end{proof}
	After the computation of the primary components of $V_{T}\cap H_A\cap H_B$, Eve has to find the $j$-invariant of the component of degree $4m$. This is explained by the next proposition.
	\begin{proposition}
		The image of a curve of bi-degree $(2,2)$ through a $ \sigma $-embedding $	 \mathbb{P}^{1}\times\mathbb{P}^{1}\to \mathbb{P}^{{m+3\choose 3}-1}$ is a curve of degree $ 4m $.
	\end{proposition}
	\begin{proof}
		In fact it is projectively equivalent to the image of a curve of bi-degree (2,2) under the map
		\[
		|\mathcal{O}_{\mathbb{P}^{1}\times\mathbb{P}^{1}}(m,m)|:\mathbb{P}^{1}\times\mathbb{P}^{1}\to \mathbb{P}^{(m+1)^{2}-1}.
		\]
		The degree of the image is $ (2,2)\cdot(m,m) =4m$.
	\end{proof}
	In conclusion, $ V_{T} \cap H_{A}\cap H_{B} $ is reducible curve of degree $ m^{3} $ with a component of degree $ 4m $. In order to break the system with this information, the eavesdropper needs to find 
	\begin{enumerate}
		\item the irreducible decomposition of $V_{3,m}^{M_{A}}\cap H_{A}\cap H_B$;
		\item the irreducible component of degree $ 4m $ and compute its $ j $-invariant.
	\end{enumerate}
\end{document}